\documentclass[runningheads,envcountsame]{llncs}
\pdfoutput=1
\usepackage[T1]{fontenc}

\title{Strided Difference Bound Matrices}
\author{%
Arjun Pitchanathan\inst{1}\orcidID{0000-0002-7301-2307} \and 
Albert Cohen\inst{2}\orcidID{0000-0002-8866-5343} \and
Oleksandr Zinenko\inst{2}\orcidID{0000-0003-1978-0222} \and 
Tobias~Grosser\inst{3}\orcidID{0000-0003-3874-6003}
}%
\authorrunning={Pitchanathan et al.}%
\institute{%
University of Edinburgh, UK,
\email{arjun.pitchanathan@ed.ac.uk}, 
\and
Google DeepMind, Paris, France
\email{\{albertcohen,zinenko\}@google.com}
\and
University of Cambridge, UK,
\email{tobias.grosser@cst.cam.ac.uk}}

\date{\today}

\usepackage{inputenc}

\usepackage[title]{appendix}

\usepackage{amsmath}
\usepackage{amssymb}

\usepackage{amsthm}
\usepackage{complexity}

\usepackage{graphicx}
\usepackage{algorithm}
\usepackage{algorithmicx}
\usepackage[noend]{algpseudocode}
\usepackage{wrapfig}
\usepackage{thmtools,thm-restate}

\usepackage{setspace}
\usepackage{todonotes}

\usepackage{hyperref}

\newcommand{\appref}[1]{\hyperref[#1]{Appendix~\ref*{#1}}}

\algnewcommand\Continue{\textbf{continue}}
\algnewcommand\Break{\textbf{break}}

\newcommand{\Rbb}{\mathbb{R}}

\newcommand{\Zbb}{\mathbb{Z}}
\newcommand{\Nbb}{\mathbb{N}}

\newcommand{\Ocal}{\mathcal{O}}

\newcommand{\Dlcm}{D_{\textrm{lcm}}}

\DeclareMathOperator{\lcm}{lcm}
\DeclareMathOperator{\MASD}{MASD}

\newcommand{\meanstddev}[2]{$#1 \pm #2$}

\begin{document}

\maketitle

\begin{abstract}
A wide range of symbolic analysis and optimization problems can be formalized using polyhedra. Sub-classes of polyhedra, also known as sub-polyhedral domains, are sought for their lower space and time complexity. We introduce the Strided Difference Bound Matrix (SDBM) domain, which represents a sweet spot in the context of optimizing compilers. Its expressiveness and efficient algorithms are particularly well suited to the construction of machine learning compilers. We present decision algorithms, abstract domain operators and computational complexity proofs for SDBM. We also conduct an empirical study with the MLIR compiler framework to validate the domain's practical applicability. We characterize a sub-class of SDBMs that frequently occurs in practice, and demonstrate even faster algorithms on this sub-class.
\end{abstract}

\section{Introduction and Motivation}

The analysis and verification of computing systems involves a variety of abstractions of the system semantics.
Among these, numerical abstractions capture arithmetic properties of system variables, supporting mathematical models of systems such as timed and hybrid automata \cite{alur1994theory,Timed,561342} and the static analysis of inductive definitions in loops and recursive programs \cite{10.1145/512760.512770}.
Many of these abstractions implement special cases of Presburger arithmetic \cite{isl} where typical decision problems are NP-hard.
The simplest special cases are non-relational, such as interval bounds $\pm x \le c$ where $x$ is a variable and $c$ is a numeric constant.
More expressive, relational cases include systems of inequalities of the form $\pm x \pm y \leq c$ known as Unit Two Variable Per Inequality (UTVPI) systems.
They form the \emph{octagon abstract domain} \cite{10.5555/832308.837141}.
While being much cheaper to operate upon than convex polyhedra \cite{10.1145/512760.512770,10.1016/j.scico.2005.02.003}, UTVPI are sufficiently expressive to represent a wide range of multi-variable problems \cite{DBLP:journals/corr/abs-cs-0703084}.

UTVPI algorithms rely on a Difference Bound Matrix (DBM) representation \cite{DBLP:conf/ifip/BerthomieuM83}, with inequalities of the form $x - y \leq c$ or $\pm x \leq c$.
DBMs are ubiquitous in formal verification \cite{10.1007/BFb0020949} and static analysis \cite{blanchet2003static}.
Other abstractions such as congruences over linear combinations of integral variables \cite{10.5555/111310.111320} capture only the lattice structure of Presburger sets but not inequalities.
The special case of congruence equalities $x \equiv r \mod d$ where $r, d$ are integral constants and $0 \le r < d$ has low complexity \cite{doi:10.1080/00207168908803778} and is often used to enhance other abstract domains \cite{10.1145/154630.154650}.

It has remained an open problem whether there are efficient algorithms for the conjunction of UTVPI and congruence constraints.
Such a domain would have numerous applications in the analysis and optimization of machine learning (ML) models.
Indeed, modern ML compilers~\cite{XLA,TVM,MLIR,triton} often use a form of Presburger representation for ML compute graphs and operations, e.g.\ to capture the data layout in memory or conversions such as reshaping and padding.

Affine expressions also arise in program transformations to leverage modern hardware, such as vectorization, fusion and thread-level parallelization.
Most of these expressions represent hyper-rectangular shapes, with occasional cases of symmetric and triangular ones (Cholesky factorization and sequence models~\cite{10.5555/3045118.3045374,DBLP:journals/corr/abs-1904-10509}), all of which can be expressed as UTVPI~\cite{10.1145/2429069.2429127}.
On the other hand, strides and block sizes resulting from (dilated) convolutions, pooling and normalization operations as well as the results of the tiling (block-wise decomposition) transformation require congruence constraints.
While some of the most advanced compiler optimizations justify the efforts to implement full-fledged Presburger arithmetic packages such as isl \cite{isl} and FPL \cite{FPL}, the majority of simpler cases call for a definition of a relational abstract domain combining UTVPI and congruences with a low-degree polynomial complexity. We also expect such a domain to be applicable to verification efforts \cite{10.1145/3527328,reinking2022formal,10.1007/978-3-031-13188-2_19} that currently rely on Presburger arithmetic libraries and SMT solvers; we present early results in Section~\ref{sec:tv}.

This paper 
considers the conjunction of inequalities represented as a DBM with single-variable congruences, a novel abstract domain we call \emph{Strided Difference Bound Matrices} (SDBM). We also study a sub-case of these, \emph{Harmonic SDBM} (HSDBM), where such congruences form a harmonic sorted chain, which is
common in congruences produced by loop tiling in high-performance code.

Although the SDBM satisfiability problem turns out to be \NP-hard, we are able to provide and algorithm that runs in $\Ocal(n m \Dlcm)$ time, where $n$ is the number of variables, $m$ is the number of constraints and $\Dlcm$ is the least common multiple of all congruence divisors. This time complexity, which is pseudo-linear in $\Dlcm$, is practical for program analysis applications. We also present an $\Ocal(n^4)$ complexity algorithm for HSDBM satisfiability.



Finally, we define a normal form for SDBM constraint systems that is computable in at most $3m + m\log(n\Dlcm) + n\Dlcm$ satisfiability checks in the general case, and $3m + m \log(n\Dlcm) + n$ checks in the harmonic case. Given two systems in normal form, we show that it only takes linear time to perform the $\mathrm{join}$ operation, producing a constraint set admitting a union of solutions, common in abstract interpretation. Moreover, we can perform an equality check based on direct comparison of normal forms.




\section{DBMs, SDBMs, and HSDBMs}

We consider sets over the integers only, i.e., subsets of $\Zbb^n$ for some $n \in \Nbb$. We first define some notation.
For $m, n \in \Nbb$, $[n]$ denotes the set $\{1, \dots n\}$, $n\Zbb$ denotes the set of integer multiples of $n$, and $m \mid n$ denotes that $m$ divides $n$. 
If $G$ is a weighted graph with no negative cycles and $u$ and $v$ are vertices in it, then $\delta_G(u, v)$ is the distance from $u$ to $v$ in $G$. $\lfloor x \rfloor_y$ refers to $x$ rounded down to the nearest multiple of $y$ smaller than or equal to $x$. If $x, y$ are vectors and $t$ a scalar, $x + t$ refers to element-wise addition.

Let us first formally recall the definition of Difference Bound Matrices (DBM) over integers and their properties~\cite{dill1990dbm,mine2001dbm} before presenting SDBM.

\label{sec:dbm}

\begin{definition}
\label{def:dbm}
A \emph{Difference Bound Matrix (DBM)} is a constraint system over variables $x_1, \dots x_n \in \Zbb$ of the form
\begin{align*}
- x_i + x_j \le c_{ij}
&& \ell_i \le x_i \le u_i
&& (i, j, c_{ij}) \in E \textrm{ and } l_i, l_j \in \Zbb
\end{align*}
where $E \subseteq [n] \times [n] \times \Zbb$ denotes the set of difference bound constraints. We will use $m = |E|$ to denote the number of such constraints.

Not all upper and lower variable bounds $\ell_i$, $u_i$ may be present. When no such variable bounds are present we call the system \emph{variable-bound-free (VBF)}; otherwise we say that the system has variable bounds.
\end{definition}
It is known that the satisfiability of DBM constraints can be determined in $\Ocal(n^3)$ time and $\Ocal(n^2)$ space \cite{DBLP:conf/ifip/BerthomieuM83,DBLP:journals/corr/abs-cs-0703084}. We now define two special cases of Presburger sets derived from DBMs by introducing additional congruence constraints.



\begin{lemma}[DBM Shifting Lemma]
\label{lem:dbm-shift}
If $x$ is a solution to a VBF DBM, then so is $x + t$ for $t \in \Zbb$, i.e., adding a constant to all variables preserves satisfiability.
\end{lemma}
\begin{proof}
All constraints are bounds on differences of variables, and the differences don't change when adding a constant to all variables.
\end{proof}
\begin{corollary}
If $S_{i,t}$ is the set of solutions to a VBF DBM such that $x_i = t$, then $S_{i,t} = \{x + t \mid x \in S_{i,0}\}$.
\end{corollary}

Given a DBM with variable bounds, we can construct a new VBF system by adding a new variable $x_0$ and converting all variable bounds $\ell_i \le x_i \le u_i$ to difference bound constraints $\ell_i \le x_i - x_0 \le u_i$.
Clearly $(x_1, \dots x_n)$ is a solution to the original system iff $(0, x_1, \dots x_n)$ is a solution to the new system. By the above lemma, the new system has a solution with $x_0 = 0$ iff it has any solution. Thus the original DBM with variable bounds is satisfiable iff the new VBF DBM is. VBF DBMs are best understood by analyzing their \emph{potential graphs}.

\begin{definition}
The \emph{potential graph of a DBM} is a weighted directed graph over vertex set $[n]$ with an edge from $i$ to $j$ of weight $c_{ij}$ for each $(i, j, c_{ij}) \in E$. The weights may be negative and the graph may contain negative cycles.
\end{definition}



\begin{lemma}
\label{lem:dbm-path}
Let $G = ([n], E)$ be the potential graph of a DBM. If $G$ has a path from vertex $u$ to $v$ of total weight $W$,
then $-x_u + x_v \le W$ for every solution $x$ to the DBM.
\end{lemma}
\begin{corollary}
\label{lem:dbm-dist-bound}
If the graph has negative cycles, then no solution $x$ exists.

If the graph has no negative cycles, then for all $u, v \in [n]$, it holds that $-x_u + x_v \le \delta_G(u, v)$.
This is useful to define a normal form of the DBM.
\end{corollary}

\begin{definition}
A \emph{path-closed DBM} is one that is satisfiable and, for all $u, v$ such that there exists a path from $u$ to $v$ in the potential graph $G$, the bound on $-x_u + x_v$ exists and is equal to $\delta_G(u, v)$.
\end{definition}

Clearly, any DBM can be brought to path-closed form by computing the distances in the potential graph, and by \autoref{lem:dbm-dist-bound}, doing so does not change the solution set. Moreover, the path-closed form has the following useful property.

\begin{lemma}[DBM Projection Lemma]
\label{lem:dbm-proj}
If a DBM is path-closed, then the projection of its solution set onto a subset of variables is equal to the
solution set of the constraints involving only those variables.
\end{lemma}

It follows that the path-closed form is the tightest constraint system with the same solution set as the original system, i.e., in a path-closed DBM there exist solutions on the surface of every inequality, so no inequality can be further tightened without changing the solution set. Moreover, if there is no constraint on some $-x_i + x_j$ then adding any upper bound on this changes the solution set. 
Finally, the following is useful to compute a complete explicit solution.

\begin{lemma}
\label{lem:dbm-vert-dist-soln}
For any vertex $u$ in the potential graph from which all other vertices are reachable, the assignment $x_v = \delta_G(u, v)$ satisfies the DBM.
\end{lemma}
Note that in this solution, $x_u = 0$.
%
%
%
%
%
We now define the new abstract domains.

\begin{definition}
\label{def:SDBM}
A \emph{Strided DBM (SDBM)} is a DBM with additional constraints
\begin{align*}
x_i \equiv r_i \mod d_i && i \in [n]
\end{align*}
where all $d_i, r_i$ are in $\Zbb$. When referring to such a system, $\Dlcm$ will denote $\lcm(d_1, \dots d_n)$. Given an SDBM, we define the \emph{underlying DBM} as the constraint system without these congruence constraints.
\end{definition}
\noindent
Note that one may encode the lack of a congruence constraint as $x_i \equiv 0 \mod 1$.

\begin{definition}
\label{def:hSDBM}
A \emph{Harmonic SDBM (HSDBM)} constraint system is an SDBM where the congruence divisors are sorted and each one divides the next, i.e., $d_1 \mid d_2 \mid \dots \mid d_n$.
\end{definition}

\section{Satisfiability}


We start by reducing the SDBM satisfiability problem to a simpler form. Firstly, let $y_i = x_i - r_i$. Then we can see that $x_i \equiv r_i \mod d_i$ iff $y_i \equiv 0 \mod d_i$. Furthermore, $-x_i + x_j \le c_{ij}$ iff $-y_i + y_j \le c_{ij} + r_i - r_j$. Thus the original SDBM
\begin{align*}
    x_i \equiv r_i \mod d_i &&
    -x_i + x_j \le c_{ij} &&
    \ell_i \le x_i \le u_i
\end{align*}
is satisfiable iff the following system is:
\begin{align*}
    y_i \equiv 0 \mod d_i \qquad
    - y_i + y_j \le c_{ij} + r_i - r_j \qquad
    \ell_i - r_i \le y_i \le u_i - r_i.
\end{align*}
Thus we reduce satisfiability of any SDBM to the satisfiability of another SDBM where all congruence constraints have remainder zero. We can further reduce satisfiability of SDBMs with variable bounds to satisfiability of VBF SDBMs.
To do this, we generalize the DBM shifting lemma to SDBMs.

\begin{lemma}[SDBM Shifting Lemma]
\label{lem:sdbm-shift}
If $x$ is a solution to a VBF SDBM, then so is $x + t\Dlcm$ for $t \in \Zbb$.
\end{lemma}
\begin{proof}
By the DBM shifting lemma (\autoref{lem:dbm-shift}), the inequality constraints continue to be satisfied. Since the scalar being added is a multiple of all the congruence divisors, the congruence constraints also continue to be satisfied.
\end{proof}
\begin{corollary}
\label{lem:sdbm-equiv}
For a given VBF SDBM with the congruence constraint on $x_n$ being $x_n \equiv 0 \mod \Dlcm$, let $S_t$ be the set of solutions such that $x_n = t$. Then $S_t = \{x + t \mid x \in S_0\}$ for $t \in \Dlcm \Zbb$. (Of course, $S_t = \varnothing$ for non-congruent $t$).
\end{corollary}
We convert SDBMs to VBF form similarly to the procedure for DBMs. Let $C$ be an SDBM with variable bounds and all remainders zero. Now create a VBF SDBM $C'$ by adding a variable $x_0$ and replacing the constant bounds $\ell_i \le x_i \le u_i$ of $C$ with inequalities $\ell_i \le x_i - x_0 \le u_i$. Then the set of solutions of $C$ is equal to the set of solutions of $C'$ such that $x_0 = 0$. Now by the above corollary, if we add the constraint that $x_0 \equiv 0 \mod \Dlcm$, then $C'$ is satisfiable iff $C$ is satisfiable.
Thus satisfiability of SDBMs with variable bounds can be efficiently reduced to satisfiability of the following simpler class of SDBMs.

\begin{definition}
\label{def:sdbm-simple}
A constraint system of the form
\begin{align*}
    x_i \in d_i \Zbb
    && -x_i + x_j \le c_{ij}
    && (i, j, c_{ij}) \in E
\end{align*}
is called a \emph{simple SDBM}. We sometimes refer to $d_i$ as the \emph{stride} of the variable $x_i$. When the system satisfies $d_1 \mid \dots \mid d_n$, we call it a \emph{simple HSDBM}.
\end{definition}

Not all SDBMs have an equivalent simple SDBM. Let us describe how to find such a simple representation if one exists.
An unsatisfiable SDBM can, of course, be represented as a simple SDBM with the same empty solution set. A satisfiable SDBM $C$ with variable bounds can never have the same solution set as a variable-bound-free SDBM $C'$, because by the shifting lemma (\autoref{lem:sdbm-shift}), for any solution $x$ of $C'$, there exists a constant $D$ such that $x + kD$ is also a solution for any integer $k$. This establishes that the solution set of $C'$ does not satisfy any variable bounds, so $C'$ has a different solution set than $C$.

A VBF SDBM $C$ with non-zero remainders admits a simple SDBM representation if there is a way to replace its congruence constraints with zero-remainder constraints while preserving the same solution set. This can be determined by computing the possible remainders of all variables modulo $\Dlcm$. By the shifting lemma, a remainder $x_i \equiv r_i \mod \Dlcm$ is possible iff there is a solution with $x_i = r_i$, which amounts to a satisfiability check.

Let $g_i$ be the GCD of all possible remainders obtained above and $\Dlcm$. By the shifting lemma, $g_i$ is the GCD of all valid values of $x_i$. To ensure that our new congruence constraint for $x_i$ does not invalidate any solutions of the original system, it is necessary and sufficient that the new divisor be a divisor of $g_i$.

To disallow any extraneous solution, we make the congruence constraint as sparse as possible. Consider the system $C'$ with congruence constraints $x_i \equiv 0 \mod g_i$ and the inequality constraints of $C$. $C$ can be represented by a simple SDBM with the same solution set iff $C$ and $C'$ have the same solution set.

\subsection{GCD-Tightening constraints}
If a DBM is unsatisfiable, repeatedly applying the following inference rule will produce a contradiction eventually.
\begin{align*}
-x_i + x_j \le c_{ij} \land -x_j + x_k \le c_{jk} \Rightarrow -x_i + x_k \le c_{ij} + c_{jk} && \text{ (path inference rule) }
\end{align*}
This is because if the DBM is unsatisfiable then a negative cycle exists, 
and in that case, repeatedly applying the above leads to an inequality of the form $0 \le c$ for some negative $c$. In an SDBM, if the underlying DBM is unsatisfiable then the above is true. However, it is possible for an unsatisfiable SDBM to have its underlying DBM be satisfiable. Consider the following example:
\begin{align*}
    x, y \in 2\Zbb &&
    1 \le x - y \le 1
\end{align*}
The inequalities on their own are clearly satisfiable over the integers. However, because both $x$ and $y$ are even, $x - y$ cannot be $1$ as required by the above system. Due to the congruence constraints, $x - y \le 1$ implies $x - y \le 0$ and similarly $1 \le x - y$ implies $2 \le x - y$, so the system is unsatisfiable.
In general, by Bézout's lemma, when $x \in a\Zbb$, $y \in b\Zbb$, then $x - y \in \gcd(a, b)\Zbb$. Thus we can always tighten bounds to a multiple of the GCD, leading to a new inference rule:
\begin{align*}
    -x_i + x_j \le c_{ij} \implies -x_i + x_j \le \lfloor c_{ij} \rfloor_{\gcd(d_i, d_j)} && \text{ (GCD-tightening rule) }
\end{align*}
We use the above to define a GCD-tight SDBM.

\begin{definition}
A \emph{GCD-tight SDBM} is one where, for all $i, j \in [n]$, we have $c_{ij} \mid \gcd(d_i, d_j)$.
\end{definition}

These two rules are still not sufficient to determine if an SDBM is satisfiable. The following system is GCD-tight, path-closed, and the inequalities are satisfiable over integers, but the system as a whole is unsatisfiable.
\begin{align}
    x &\equiv 0 \mod 4 \cdot 5 \nonumber && 0 \le y - x \le 5 \\ \label{eq:sdbm-counter}
    y &\equiv 0 \mod 5 \cdot 7 && 20 \le x - z \le 24 \\ 
    z &\equiv 0 \mod 4 \cdot 7 && \nonumber 21 \le y - z \le 28 \nonumber
\end{align}
To see that it is unsatisfiable, reparameterize the solution as $(c + a, c + b, c)$; this vector is a solution to the congruences iff
\begin{align*}
    c \equiv a &\mod 4 \cdot 5 &
    c \equiv b &\mod 5 \cdot 7 &
    c \equiv 0 &\mod 4 \cdot 7
\end{align*}
which by the general Chinese remainder theorem~\cite{ore1952crt} has solutions iff
\begin{align*}
    a &\equiv b \mod 5 &
    a &\equiv 0 \mod 4 &
    b &\equiv 0 \mod 7.
\end{align*}
Since the solution is of the form $(a, b, 0) + c$, it satisfies the inequalities iff $(a, b, 0)$ does, by \autoref{lem:dbm-shift}. Thus the inequalities hold iff
\begin{align*}
    0 &\le b - a \le 5 &
    20 &\le a \le 24 &
    21 &\le b \le 28.
\end{align*}
Due to the congruence constraints we have $a \in \{20, 24\}$, $b \in \{21, 28\}$, and $b - a \in \{0, 5\}$, which cannot be satisfied simultaneously, so the SDBM is unsatisfiable.
For the case of HSDBMs however, path-closure and GCD-tightening suffice.

\subsection{Satisfiability for HSDBMs in $O(n^4)$ time}
By the earlier discussion, we can assume that the given HSDBM is simple. In this case, path-closure and GCD-tightening are sufficient to determine satisfiability. To show this, we prove a projection lemma for HSDBMs; while the general projection lemma for DBMs (\autoref{lem:dbm-proj}) does not apply to HSDBMs, it does hold when the subset of variables chosen forms a suffix.
We will call an HSDBM path-closed when its underlying DBM is path-closed.

\begin{restatable}{lemma}{projHSDBMSuffix}
\label{lem:proj-hr-suffix}
Let $H$ be a path-closed, GCD-tight VBF HSDBM. If $S_{k:n}$ is the projection of the solution set of $H$ onto $x_k, \dots x_n$, then $S_{k:n}$ is equal to the set of solutions to the inequalities and congruence constraints involving only $x_k, \dots x_n$.
\end{restatable}
\begin{proof}
Suppose $(p_{k+1}, \dots p_n) \in S_{k+1:n}$. We show that there exists a $p_k$ such that $(p_k, \dots p_n) \in S_{k:n}$. By substituting $p_{k+1}, \dots p_n$ into the system, we obtain bounds of the form $p_i - c_{ki} \le x_k \le p_i + c_{ik}$ for $k < i \le n$ on $x_k$ when the corresponding inequalities exist. If none of the lower bounds exist or none of the upper bounds exist, then we can definitely find a multiple of $d_k$ satisfying these bounds to assign to $x_k$.

Otherwise, if at least one upper bound and one lower bound is produced, then the set of $x_k$ satisfying the inequalities is $[\max_i(p_i - c_{ki}), \min_i(p_i + c_{ik})]$, which is of the form $[p_i - c_{ki}, p_j + c_{jk}]$ for some $i, j \in \{k+1, \dots n\}$. This interval is non-empty by the DBM projection lemma (\autoref{lem:dbm-proj}).
%
%

Now note that $p_i \in d_i \Zbb \subseteq d_k \Zbb$ by the harmonic property and similarly $p_j \in d_k \Zbb$. Also, $c_{ki} \in \gcd(d_k, d_i) \Zbb = d_k \Zbb$ by path-closure and the harmonic property; similarly,  $c_{jk} \in d_k \Zbb$. So both the endpoints lie in $d_k \Zbb$ and therefore it certainly contains a multiple of $d_k$. Repeating this, we can extend any point in $S_{k:n}$ into a point in $S_{1:n}$, a full solution to the whole system.
\end{proof}
\begin{corollary}
A path-closed GCD-tight HSDBM is satisfiable.
\end{corollary}
\begin{proof}
$S_{n:n} = d_n\Zbb \ne \varnothing$ is the projection of the solution set onto $x_n$.
\end{proof}

This forms the basis of the \textsc{SolveHSDBM} algorithm in \autoref{fig:sat} to decide the satisfiability of HSDBMs: first, obtain the path-closure of the inequalities by running the Floyd-Warshall algorithm~\cite{clrs} on the potential graph, then GCD-tighten all inequalities, and repeat these two steps until a fixpoint or contradiction is reached, at which point we know whether the system is satisfiable.
\begin{figure}[h!t]
\begin{minipage}[t]{.49\textwidth}
\small
\begin{algorithmic}[1]
\Function{SolveHSDBM}{$E, d$}
\State Path-close inequalities $E$
\While{$(E, d)$ not a fixpoint}
\State Set every $c_{ij}$ in $E$ to \\
\hfill $\lfloor c_{ij} \rfloor_{\gcd(d_i, d_j)}$
\If{negative cycles in $E$}
    \State \Return $\bot$
\EndIf
\State Compute all pairs of distances
\State Set every $c_{uv}$ to $\delta_E(u, v)$
\EndWhile
\State \Return \texttt{SAT}
\EndFunction
\end{algorithmic}
\end{minipage}
\hfill
\begin{minipage}[t]{.48\textwidth}
\small
\begin{algorithmic}[1]
\Function{SolveSDBM}{$E, d$}
    \If{no integral solution to $E$}
        \State \Return $\bot$
    \EndIf
    \State $p \gets $ an integral solution to $E$
    \State $\Dlcm \gets \lcm(d_1, \dots d_n)$
    \State $\ell \gets p - n\Dlcm$
    \State $u \gets p + n\Dlcm$
    \For{$i \in [n]$}
        \State $u_i \gets \lfloor u_i \rfloor_{d_i}$
    \EndFor
    \While{fixpoint not reached} \label{alg-line:fixpoint-loop-start}
        \For{$(i, j, c_{ij}) \in E$}
            \If{$u_j < \lfloor u_i + c_{ij} \rfloor_{d_j}$} \label{alg-line:if-ui}
                \State $u_j \gets \lfloor u_i + c_{ij} \rfloor_{d_j}$
                \If{$u_j < \ell_j$}
                    \State \Return $\bot$
                \EndIf
            \EndIf
        \EndFor
    \EndWhile \label{alg-line:fixpoint-loop-end}
    \State \Return $u$
\EndFunction
\end{algorithmic}
\end{minipage}
\label{fig:sat}
\caption{HSDBM and SDBM satisfiability.}
\end{figure}


\begin{lemma}
\label{lem:induced-fw}
Let $G = (V, E)$ be a transitively closed graph with no negative cycles, i.e.\ whenever there is a path from $u$ to $v$, there is an edge from $u$ to $v$ of weight $\delta_G(u, v)$. Let $U \subseteq V$. Now let $F$ be a copy of $E$ in which we have decreased the weights of some edges that go from one vertex in $U$ to another in $U$. Finally, let $H = (V, F)$.

Suppose that $H$ has no negative cycles. Then for any vertices $u$ and $v$ in $U$ with a path from $u$ to $v$, there is a shortest path from $u$ to $v$ that never leaves $U$.
\end{lemma}
\begin{proof}
Let $p = (p_1, \dots p_k)$ be the vertices of a path starting and ending in $U$ and with all the intermediate vertices lying outside $U$. Let $W$ be the weight of $p$ in $H$ and let $c$ be the weight of the edge from $p_1$ to $p_k$ in $H$. Then $W \ge \delta_G(p_1, p_k)$ because only edges that stay within $U$ decreased, and $\delta_G(p_1, p_k) \ge c$ because the edge in $G$ had weight equal to $\delta_G(p_1, p_k)$ and it can only have decreased or stayed the same in $H$. Thus the path $p$ cannot have weight less than the weight of the direct edge in $H$.

For a general path that goes in and out of $U$ repeatedly, we can always replace all sections of the path that go outside and come back in with the direct edges that stay in $U$, to obtain a path within $U$ whose weight is at most that of the original path. Thus for any start and end point in $U$, the shortest path that stays in $U$ has weight equal to the shortest path in the entire graph $H$.
\end{proof}

\begin{theorem}
\label{thm:hsdbm-runtime}
\textsc{SolveHSDBM} in \autoref{fig:sat} terminates in $O(n^4)$ time.
\end{theorem}

\begin{proof}
We will view the algorithm as operating on the potential graph; all modifications to $c_{ij}$ then become modifications to the edge weights. We will show that at most $n - 1$ repetitions are needed for fixpoint. We prove that after the $i$th application of GCD tightening, all edges between vertices in $\{v_i, \dots v_n\}$ will stay multiples of $d_i$ for the rest of the algorithm. We prove this by induction. The base case for $i = 1$ is true since when all edges are multiples of $d_1$, path-closure cannot change this divisibility, and GCD-tightening will not change this either.

Now assume it to be true for $i$; we will show it for $i + 1$. The $i + 1$-th application of GCD tightening only decreases edge weights between vertices in $U = \{v_{i+1}, \dots v_n\}$, by the induction hypothesis. Now we want to analyze how path closure affects the edge weights in the subgraph induced by $U$.
After tightening, all edges in the subgraph are multiples of $d_{i+1}$, so distances between nodes in the subgraph are also multiples of $d_{i+1}$ by 
\autoref{lem:induced-fw}.
Thus path closure does not affect divisibility at this step. Therefore, subsequent GCD-tightening does not affect it either. Repeated applications of these preserve the property.

Thus the $n$th application of GCD tightening does nothing since there are no edges in the graph induced on the single vertex $v_n$ for $i = n$. Therefore, neither does the subsequent application of path-closure. Thus, fixpoint is achieved after $n - 1$ runs of GCD-tightening and path-closure.
\end{proof}

\subsection{Satisfiability for SDBMs in $O(n m \Dlcm)$ time}
\label{sec:sdbm-sat}

Extending work by Lagarias~\cite{lagarias1985simultaneous}, it can be shown that the SDBM satisfiability problem is \NP-hard (\autoref{thm:sdbm-sat-np-hard}), so no polynomial-time algorithm is likely to exist. In program analysis applications, the inequality coefficients can be large, so we would like an algorithm that runs in polynomial time in the representation size of these coefficients. On the other hand, in these applications, the congruence divisors are typically small, so we are willing to let the algorithm be polynomial in the \emph{values} of these, i.e., pseudo-polynomial in these. In fact, these divisors typically share many common factors, so that their LCM is not much bigger than the divisors. We present an algorithm that is pseudo-linear in the LCM.

The intuition for the algorithm comes from the following extensions of our inference rules to upper bounds $x_i \le u_i$ on the variables.
\begin{align*}
    x_i \le u_i &\implies x_i \le \lfloor u_i \rfloor_{d_i} \\
    x_i \le u_i \land -x_i + x_j \le c_{ij} &\implies x_j \le u_i + c_{ij}
\end{align*}

Suppose we have an SDBM with all variables bounds present and we keep applying these rules. Then we either obtain a contradiction $u_i < \ell_i$, or a fixpoint. At the fixpoint it holds that $u_i \in d_i \Zbb$ and moreover $u_j \le u_i + c_{ij}$. So in fact, $u$ becomes a solution to the SDBM. Each successful application of an inference rule reduces the gap $u_i - \ell_i$ between some upper bound and lower bound. If this difference becomes negative, a contradiction is obtained and the algorithm halts.

So the worst-case runtime of this method depends on the sum of the gaps $u_i - \ell_i$ between the upper bounds and the lower bounds, which could naively be exponential in the representation size of the constraint system. To avoid this worst-case scenario, we reduce the satisfiability of SDBMs to the satisfiability of SDBMs with variable bounds where the gap between the upper and lower bound is at most $2n\Dlcm$ for each variable.

For a matrix $A$, let $\MASD(A)$ be the maximum absolute determinant among all square submatrices of $A$. A standard fact~\cite{cgst1986sensitivity} in the theory of integer programming is that if $P \subseteq \Rbb^n$ is a polyhedron, $P \cap \Zbb^n$ is non-empty, and $x$ is in $P$, then there exists a point $y$ in $P \cap \Zbb^n$ such that $\|x - y\| \le n \MASD(A)$.
We slightly generalize this to obtain the following lemma.

\begin{restatable}{lemma}{lemSdbmSolInnD}
\label{lem:sdbm-sol-in-nD}
Let $S = \{x \in \Rbb^n \mid Ax \le b\}$ be a non-empty polyhedron. Let $L$ be the set of solutions to some single-variable congruence constraints such that $S \cap L \ne \varnothing$, and let $\Dlcm$ be the LCM of the congruence divisors of $L$. Let $y \in S$. Then there exists a solution $x \in S \cap L$ such that $\|x - y\|_\infty \le n\Dlcm \MASD(A)$.
\end{restatable}

Moreover, we show that $\MASD(A) = 1$ for DBMs, making the bound $n\Dlcm$.


\begin{restatable}{lemma}{lemMASD}
Let $A$ be an $m \times n$ matrix where each row has exactly one $+1$ and one $-1$. Then $\MASD(A) = 1$.
\end{restatable}

We defer the proofs to \appref{app:lemmas}. These lemmas allow to solve an SDBM by first finding any integral solution $p$ to the inequalities and adding constant bounds on the variables to lie within a box of side length $2n\Dlcm$ centered on $p$, then applying the above inference rules until reaching a contradiction or fixpoint.

In \textsc{SolveSDBM} in \autoref{fig:sat}, we first GCD-tighten all the upper bounds and then look for opportunities to apply the path-closure inference rule to the upper bounds, by checking each difference bound. Whenever an upper bound decreases due to path-closure, we immediately apply the GCD-tightening rule to it. It takes $O(m)$ time to look over all edges. Since each variable's upper and lower bounds differ by $2n\Dlcm$, there can be at most $2n^2\Dlcm$ steps of such tightening, for an overall runtime of $\Ocal(n^2m\Dlcm)$.

We have to process the edge $(i, j, c_{ij})$ once in the beginning. After that, we only have to process it again when the RHS of the if-condition on line \ref{alg-line:if-ui} changes, i.e., only when $u_i$ decreases. So we can replace lines \ref{alg-line:fixpoint-loop-start}-\ref{alg-line:fixpoint-loop-end} with the following.

{\small%
\begin{algorithmic}[1]
    \State $\mathtt{dirty} \gets [n]$
    \While{$\mathtt{dirty} \ne \varnothing$}
    \State Pick $i \in \mathtt{dirty}$
    \State Remove $i$ from $\mathtt{dirty}$
    \For{$(i, j, c_{ij}) \in E$}
        \If{$u_j < \lfloor u_i + c_{ij} \rfloor_{d_j}$}
            \State $u_j \gets \lfloor u_i + c_{ij} \rfloor_{d_j}$
            \If{$u_j < \ell_j$}
                \State \Return $\bot$
            \EndIf
            \State Add $j$ to $\mathtt{dirty}$
        \EndIf
    \EndFor
    \EndWhile
\end{algorithmic}}
\noindent
Here $u_i$ can decrease at most $2n\Dlcm$ times since after that it will go below the lower bound $\ell_i$ and produce a contradiction. Each time $u_i$ decreases, we check all edges that go out from $i$, as these are the edges that might use the reduced value of $u_i$. Thus if $o_i$ is the number of edges leaving $i$, then the time complexity of this more careful implementation is $\sum_i \Ocal(n\Dlcm o_i) = \Ocal(n\Dlcm m)$ since $\sum_i o_i = m$. 



\section{HSDBM Normalization}

We consider normalization for satisfiable systems; if a system is unsatisfiable we  normalize it by setting it to some canonical unsatisfiable system. We first normalize the inequalities and then the congruence constraints.
\begin{definition}
\label{def:ineq-norm}
An \emph{inequality-normalized (H)SDBM} is one where for any bound $-x_i + x_j \le c$, if it holds in the solution set that $-x_i + x_j \le d$, then $c \le d$, i.e. the bound in the system is the tightest valid bound.
\end{definition}
Note that any two SDBMs with the same solution sets will have the same normalized inequalities, since this depends only on the solution set and not on the form of the initial constraint system. The above definition is equivalent to saying that every bound $-x_i + x_j \le c$ has a solution that makes it tight, and whenever a bound does not exist that expression can take arbitrarily large values in the solution set. Also, an inequality-normalized system is always path-closed and GCD-tight since no such tightening inference rules can decrease any bound.

We previously showed that path-closure and GCD-tightening are not sufficient to check satisfiability of SDBMs. Thus, we do not expect these to be sufficient for inequality normalization either. One might hope that it is enough for HSDBMs, but in fact it is not the case either. Consider the following example.
\begin{align*}
    -1 \le x - y \le 1 &&
    -1 \le x - w \le 0 &&
    0 \le x - z \le 1 &&
    x, y \in \Zbb \\
    0 \le w - z \le 2 &&
    -1 \le y - w \le 0 &&
    0 \le y - z \le 1 &&
    z, w \in 2\Zbb
\end{align*}
It is obviously GCD-tight and path-closed. But all constraints are not as tight as possible.
%
Note that $w$ is either $z$ or $z + 2$. If $w = z$ then $x - z = x - w = 0$, and similarly $y - z = 0$, implying $x - y = 0$. Otherwise $w = z + 2$, then $x - w = x - z = 1$ and $y - z = 1$, so $x - y = 0$ again, yielding tighter inequalities $0 \le x - y \le 0$. Therefore, we need to do more for inequality normalization.



First, let us consider variable-bound-free systems. Suppose the system has an inequality $-x_i + x_j \le c$ and we want to check if replacing it by $-x_i + x_j \le b$ for $b < c$ excludes any solutions. This is equivalent to asking if there are any solutions with $b + 1 \le -x_i + x_j \le c$, which is a single satisfiability check. We can thus binary search over the values of $b$ to find the minimum valid one, to obtain the tightest form of the inequality. We now establish a bound on the range of such $b$ values over which we have to search.

By the projection lemma (\autoref{lem:dbm-proj}), path-closing the underlying DBM brings it to normal form. Therefore every difference bound $-x_i + x_j \le c_{ij}$ has an integral point $y$ satisfying all the inequalities and such that $-y_i + y_j = c_{ij}$. By \autoref{lem:sdbm-sol-in-nD}, there exists a solution $z$ to the whole system with $z_j \ge c_{ij} - n\Dlcm$ and $z_i \le c_{ij} + n\Dlcm$, so that $-z_i + z_j \ge c_{ij} - 2n\Dlcm$.
Therefore, the tightest version of the inequality has a bound that is tighter by at most $2n\Dlcm$. The binary search then takes at most $3 + \log(n\Dlcm)$ steps. Inequality normalization thus takes at most $m(3 + \log(n\Dlcm))$ emptiness checks.

Now consider systems with variable bounds, still with remainder zero congruence constraints. Path-closure and GCD-tightening are sufficient to normalize these, by converting them to VBF form and applying the following lemma.
\begin{lemma}
If an HSDBM is path-closed and GCD-tight then all inequalities involving $x_n$ are tight. If a bound on $-x_i + x_n$ is missing then $-x_i + x_n$ is unbounded in that direction, and similarly for bounds on $-x_n + x_i$.
\end{lemma}
\begin{proof}
First, we show it for bounds of the form $-x_n + x_i \le c_{ni}$. Suppose all such bounds exist. Then the point with $x_n = 0$ and $x_i = c_{ni}$ for $i < n$ is a solution. By GCD-tightening and the harmonic property, it satisfies the congruences. By path-closure, we have $c_{nj} \le c_{ni} + c_{ij}$, so it satisfies the inequalities.

Now consider the case where all bounds do not exist. Let $R$ be the set of variables that have a bound on $-x_n + x_i$ and let $\overline{R}$ be its complement. Set $x_i = c_{ni}$ as before, for variables in $R$, except $x_n$ which we set to zero. Now we find a way to fill in the values of $x_j$ in $\overline{R}$. Note that there can be no bound of the form $-x_i + x_j \le c_{ij}$ for $x_i \in R, x_j \in \overline{R}$ because then by path-closure we would have a bound $-x_n + x_j \le c_{ni} + c_{ij}$ which contradicts $x_j \in \overline{R}$.

Thus, assigning values to variables in $R$ can impose lower bounds on variables in $\overline{R}$, but not upper bounds. Since the whole HSDBM is non-empty we can find a solution $y$ to the subsystem of constraints that only involve variables in $\overline{R}$. Moreover, $t + y$ is a solution for any real $t \in D_{\overline{R}}\Zbb$ where $D_{\overline{R}}$ is the LCM of the congruence divisors of variables in $\overline{R}$. By making $t$ sufficiently large, $t + y$ satisfies the lower bounds imposed by substituting values for $R$ variables. Thus we have a solution making all the bounds $c_{ni}$ tight. Also, by increasing $t$ we can make the variables in $\overline{R}$ arbitrarily large so these are unbounded above.

To prove the case of bounds on $-x_i + x_n$, negate all the variables so that bounds on $-x_n + x_i$ become bounds on $-x_i + x_n$ and vice versa. Now we can apply the same proof as above.
\end{proof}

Therefore, to normalize a satisfiable HSDBM with variable bounds, we:
\begin{enumerate}
    \item Convert the system into VBF form,
    \item Bring the converted system into path-closed and GCD-tightened form,
    \item Convert the system back to a form with variable bounds, and
    \item Binary search on the remaining inequalities to normalize them.
\end{enumerate}
If the system was not satisfiable, we would find out at step 2, at which point we can normalize the system by setting it to some canonical unsatisfiable HSDBM.

Let us now consider how to congruence-normalize simple HSDBMs, so the normal form's congruence constraints must have remainder zero.

\begin{definition}
A \emph{congruence-normalized VBF HSDBM} where the congruence constraint system implies all other valid congruence constraint systems for that solution set, i.e., an HSDBM with congruence divisors $d^*_1, \dots d^*_n$ is congruence-normalized if for all HSDBMs with the same solution set having congruence divisors say $d_1, \dots d_n$, it holds that each $d_i \mid d^*_i$.
\end{definition}
Note that the above definition depends only on the solution set of a system, and so the normalized congruence system of any two systems having the same solution set will be the same. 
In a simple HSDBM, $x_n$ can always take all values in $d_n \Zbb$ by the shifting lemma (\autoref{lem:sdbm-shift}), so any system with the same solution set will have the same congruence $d_n$ for $x_n$. Therefore, $d^*_n = d_n$. We now normalize the remaining congruences iteratively, starting from $x_{n-1}$ and going downwards. Suppose that we already computed $d^*_{i+1}, \dots d^*_n$ and we want to compute $d^*_i$.

Note that for any valid congruence system it holds that $d_i \mid d_{i+1} \mid d^*_{i+1}$ by the harmonic property and congruence normalization. 
Thus $d^*_i$ is the maximum of all $d_i \mid d^*_{i+1}$ such that $x_i \in d_i \Zbb$ holds in the solution set.
By the projection lemma (\autoref{lem:proj-hr-suffix}), we can reduce this to finding the largest possible divisor for $x_1$ in a given constraint system with divisors $d_1, \dots d_n$. As shown above we only need to consider divisors $m \mid d_2$. For it to be a valid divisor, it also needs to not be so dense as to allow additional solutions; we ensure this by mandating that $d_1 \mid m$. Note that the greatest divisor cannot be a non-multiple of $d_1$ anyway, since if $m$ is a valid congruence for $x_1$ then so is $\lcm(d_1, m)$. 

\begin{theorem}
\label{lem:hr-mod-sparse}
Let $H$ be an inequality-normalized simple HSDBM. Let $L$ be the set of $m \in \Nbb$ such that $d_1 \mid m \mid d_2$ and for any solution $x$ of $H$, it holds that $x_1 \in m\Zbb$. We are interested in the sparsest possible congruence divisor, $\max L$. Let $g$ be the GCD of all $c_{i1}$ and $c_{1i}$, and let $q = \gcd(g, d_2)$.

Then $\max L$ is either $d_1$ or $q$. Moreover, it is $q$ iff a specific other HSDBM $H'$ is unsatisfiable, where the constraint system $H'$ can be computed in linear time from the system $H$.


\end{theorem}



\begin{proof}
We first show that for any $r \in L$, $r \mid g$.
Suppose not, then without loss of generality, $r$ does not divide some $c_{i1}$. Since the system is inequality normalized, it has some solution satisfying $x_1 = x_i + c_{i1}$. But since $x_i \in d_i\Zbb \subseteq r\Zbb$ and $c_{1i} \notin r\Zbb$, we have $x_1 \notin r\Zbb$, so $r \notin L$ which is a contradiction. So this case is impossible and we have $r \mid g$. Since $r \mid d_2$, we have $r \mid \gcd(g, d_2) = q$. Thus, $\max L \mid q$. If $q = d_1$, we are done and $\max L = d_1$.

Otherwise, let $q \ne d_1$. We now show that either $q \mid \max L$, implying $\max L = q$, or $\max L = d_1$.
Let $S_{2:n}$ be the projection of the solution set onto $x_2, \dots x_n$. For now, assume that all constraints in the system exist. Then every assignment $(p_2, \dots p_n) \in S_{2:n}$ implies constraints of the form $p_i - c_{1i} \le x_1 \le p_i + c_{i1}$. So the set of possible $x_1$ values for this assignment is
$\bigcap_{i=2}^n [p_i - c_{1i}, p_i + c_{i1}] \cap d_1\Zbb$.
This set is non-empty by the definition of $S_{2:n}$. Since $q \mid d_2 \mid p_i$ for all $i \ge 2$ and $q$ divides all the coefficients $c_{1i}$ and $c_{i1}$, all interval endpoints are multiples of $q$. Therefore the endpoints of the intersection are also multiples of $q$. Since $d_1 \mid q$, if the intersection contains more than one element then it definitely contains two adjacent multiples of $d_1$, implying $\max L = d_1$. Otherwise, if the intersection contains exactly one element, that element is surely a multiple of $q$.

Thus, $q \mid \max L$ if for all points in $S_{2:n}$, the intersection of the intervals is a singleton. Otherwise, $\max L = d_1$. The intersection of some intervals is a singleton iff the right endpoint of some interval equals the left endpoint of some interval, possibly the same one (\autoref{lem:interval-cap}). So we have to check whether, for every valid assignment in $S_{2:n}$, some two intervals $[x_i - c_{1i}, x_i + c_{i1}]$ and $[x_j - c_{1j}, x_j + c_{j1}]$ intersect only at their endpoints, i.e., there always exist some $i, j \in \{2, \dots n\}$ such that $x_i + c_{i1} = x_j - c_{1j}$, i.e., $-x_j + x_i = c_{i1} + c_{1j}$. Note that by path closure, if $x_2, \dots x_n \in S_{2:n}$, then it already holds that $-x_j + x_i \le c_{ji} \le c_{j1} + c_{1i}$. So it is only left to check whether 
$\forall x_2, \dots x_n \in S_{2:n},\; \exists i, j \in \{2, \dots n\},\; - x_j + x_i \ge c_{j1} + c_{1i}$.
By logically negating twice, this is equivalent to
$\lnot \exists x_2, \dots x_n \in S_{2:n},\; \forall i, j \in \{2, \dots n\},\; - x_j + x_i < c_{j1} + c_{1i}$.
The strict inequality is equivalent to $-x_j + x_i \le c_{j1} + c_{1i} - 1$ since all variables are integers. By the HSDBM projection lemma (\autoref{lem:proj-hr-suffix}), a vector belongs to $S_{2:n}$ iff it satisfies the constraints on those variables in the HSDBM. Thus the condition above can be checked using a single HSDBM satisfiability check.

If some of the $c_{1i}$ or $c_{i1}$ bounds did not exist then the corresponding intervals in the intersection would have ranged till infinity on that side. Still, the same conclusion holds: $\max L \ne d_1$ iff the intersection is a singleton, meaning that some two finite endpoints have to coincide, and the rest of the proof proceeds the same way. Whenever some $c_{1i}$ or $c_{i1}$ does not exist we simply do not add any of the bounds in the constructed system that depend on that bound.
\end{proof}

\subsubsection{Generalizing to HSDBMs with variable bounds.}
When variable bounds exist, it is possible for a variable to take only a single value, in which case any congruence divisor is valid and the sparsest congruence constraint is not well-defined. In this case, in inequality-normalized form, the variable will have upper and lower bounds equal, so we can immediately detect this case by looking at the variable bounds. When this happens, we first eliminate these variables by substituting in the single value that they can take. We then compute the congruence normalization of the resulting system, then add back the eliminated variables and give them some canonical congruence constraint that preserves the harmonic property. For example, use $x_1 \equiv 0 \mod 1$ if it is the first variable and use the divisor of the previous variable otherwise.

We now consider congruence normalization of systems with variable bounds where every variable takes at least two values.

\begin{lemma}
\label{lem:sdbm-var-bound-vbf-cong-normal}
Let $C$ be an SDBM with variable bounds, where each variable takes at least two values. Let $C'$ be the system converted into VBF form with the added variable $x_{n+1}$ having divisor $D$, with $\Dlcm \mid D$. Let $e_1, \dots e_n$ be the normalized congruence divisors of the converted system, and let $d^*_1, \dots d^*_n$ be the sparsest congruences for the original system. Then $\forall i, e_i = \gcd(D, d^*_i)$.
\end{lemma}
\begin{proof}
Let $S$ be the set of values $x_i$ takes in $C$ and let $T$ be the set of values it takes in  $C'$. Then $T = \{x + tD \mid x \in S, t \in \Zbb\}$ by the shifting lemma (\autoref{lem:sdbm-shift}). The sparsest congruence divisor for $S$ is the GCD of all elements in $S$, which we call $g$. Similarly, the sparsest congruence divisor for $T$ is the GCD of all elements in $T$, which is equal to $\gcd(g, D)$ since for any $a$, $\gcd_{t \in \Zbb}(a + tD) = \gcd(a, D)$.
\end{proof}


\begin{lemma}
\label{lem:hsdbm-bounds-normal-dn}
In an HSDBM with variable bounds where $x_n$ takes at least two possible values, the sparsest possible congruence divisor for $x_n$ is $d_n$.
\end{lemma}
\begin{proof}
Convert the system to VBF form. Let $x_{n+1}$ be the variable added for the conversion. By the projection lemma (\autoref{lem:proj-hr-suffix}), the set of valid values of these two variables is the set of constraints involving only them. The set of valid values of $x_n$ in the original system is the set of values of $x_n$ in the converted system with $x_{n+1} = 0$, and is therefore the set of multiples of $d_n$ within the variable bounds of $x_n$. Thus the sparsest congruence divisor for $x_n$ is still $d_n$. 
\end{proof}

We now show how to compute the sparsest congruence constraints.
\begin{theorem}
Given an HSDBM with variable bounds where every variable takes at least two values, 
the sparsest congruence constraints are equal to sparsest constraints for the system after converting to VBF form.
\end{theorem}
\begin{proof}

We convert the system to VBF form by adding a variable $x_{n+1}$ with congruence divisor $d_{n+1} := d_n$. We then compute its congruence normalization to obtain divisors $e_1 \mid \dots \mid e_n \mid e_{n+1}$. 
Let $d^*_1 \mid \dots \mid d^*_n$ be the true sparsest congruences for the input system. Then $e_i = \gcd(d_{n+1}, d^*_i)$ by \autoref{lem:sdbm-var-bound-vbf-cong-normal} and $d_{n+1} = d_n = d^*_n$ by \autoref{lem:hsdbm-bounds-normal-dn}. Hence $e_i = \gcd(d^*_n, d^*_i) = d^*_i$ since $d^*_i \mid d^*_n$.
\end{proof}

\subsubsection{Generalizing to arbitrary congruence constraints.}
For HSDBMs with arbitrary congruence constraints, we can find any solution and shift the system so that the origin becomes a solution. Then all valid congruence constraints have remainder zero since there is a solution at the origin. Computing the sparsest possible congruence for this system and performing the inverse shift therefore gives us the sparsest possible congruence for the original system.

\section{Operations for Abstract Interpretation}
We introduce intersection, equality, inclusion, and join operations for (H)SDBMs, completing the set of operations typically required for abstract interpretation.

\subsubsection{Intersection.}
To intersect, we just take the tighter of the bounds on each $-x_i + x_j$ and of the variable bounds.

\subsubsection{Equality.}
We check if two HSDBMs have equal solution sets by checking if their normal forms are equal. For simple SDBMs, we first check if their normalized inequalities are equal, then compare congruences: given an SDBM $C$, $D \in \Nbb$ such that all $d_i \mid D$, and $r \in \{1, \dots D - 1\}$, there exists a solution with $x_i \equiv r \mod D$ iff there exists one with $x_i = r$, by the shifting lemma (\autoref{lem:sdbm-shift}).

Now given two VBF SDBMs, let $D$ be the LCM of their congruence divisors. Collecting which values modulo $D$ each variable can take in each system takes $2nD$ satisfiability checks. If both are equal and the normalized inequalities are also equal then both systems have equal solution sets. Otherwise, they do not.

Now given two SDBMs with variable bounds, we again set $D$ to be the LCM of the congruence divisors and inequality normalize both, then convert them to VBF form using the same congruence divisor $D$ for the added variable.
The two original systems are equivalent iff the converted systems are, and we know how to check equality of solution sets for VBF SDBMs.

\subsubsection{Inclusion.}
We can check for inclusion using intersection and equality checks since for sets $A$ and $B$, we have $A \subseteq B$ iff $A \cap B = A$.

\subsubsection{Join.}
Given two SDBMs in normal form, the system with the smallest solution set that encompasses both the inputs' solution sets is the system that takes the looser of the two bounds on each $-x_i + x_j$. When one of the systems has no bound, the result should have no bound either. This follows from \autoref{def:ineq-norm}.

For the congruences of the joined system, we compute the congruence normalization of both the input systems and for each variable, take the sparsest congruence constraints that encompass both. Say the two constraints are $x \equiv r_1 \mod q_1$ and $x \equiv r_2 \mod q_2$. Let $p$ be any solution to these two constraints, then the two constraints are equivalent to $x - p \equiv 0 \mod q_1$ and $x - p \equiv 0 \mod q_2$ respectively. The sparsest constraint that holds for $x - p$ satisfying either one of these constraints is $x - p \equiv 0 \mod \gcd(q_1, q_2)$, i.e., $x \equiv p \mod \gcd(q_1, q_2)$.



\section{Empirical Study}

The goal of this study is to demonstrate the \emph{suitability} of SDBM for program representation and analysis. To this end, we instrumented several optimizing compilers that use polyhedral domains internally and analyzed those domains. Evaluating the compilation time or the run time of the compiled program is beyond the scope of the study as it requires additional engineering to compete with highly-optimized Presburger arithmetic libraries~\cite{isl,FPL}.

\subsection{Methodology}

We instrumented the following compilation and analysis projects.
\begin{itemize}
    \item The MLIR compiler infrastructure~\cite{MLIR}, widely used in production to support domains ranging from machine learning compilers to hardware synthesis. We used MLIR version \texttt{llvmorg-18-init-16246-g4daea501c4fc} (Jan 5, 2024) and compiled the test suite provided with the project using \texttt{ninja check-mlir}. We collected statistics from 2176 compiler invocations. Some tests feature multiple invocations.
    \item The Polygeist CUDA-to-OpenMP cross-compiler~\cite{polygeist} based on the archived artifact~\cite{polygeist_openmp}. We compiled 17 benchmarks from the CUDA subset of the Rodinia suite~\cite{rodinia} accepted by Polygeist with the same 7 configurations as~\cite{polygeist_openmp}.
    \item The PPCG polyhedral compiler~\cite{ppcg} version \texttt{0.09.1} (Apr 2, 2023, most recent release). We compiled 30 benchmarks from the Polybench/C benchmark suite version \texttt{4.2.1}~\cite{polybench} using \texttt{ppcg --target=c --openmp --tile} to enable autoscheduling, parallelization and tiling.
 \end{itemize}

MLIR and MLIR-based Polygeist were instrumented to intercept the creation of affine expressions and sets bounded by such expressions as well as (integer) emptiness checks of these sets. For each expression and set, we verified whether it can be expressed as a (H)SDBM. We say that an expression can be  In MLIR, unique expressions are reused so that the collected statistics reflect unique SDBM objects that existed throughout the execution of the test.
PPCG, and its underlying isl library~\cite{isl}, were instrumented to check if the following objects fit (H)SDBM: affine constraints, convex sets, unions thereof and unions of non-convex sets in multiple vector spaces. We collected all such objects at several moments in the compilation process: after constructing the initial representation, after performing dependence analysis, before and after scheduling, and just before final code generation.











\subsection{Prevalence of SDBMs}

\subsubsection{MLIR.}
Out of 2176 test cases, 1264 (58.1\%) construct affine expressions throughout their lifetime. The following analysis focuses only on those.
Overall, 96.3\% of affine expressions and 95.6\% of integer sets (we consider MLIR multidimensional affine maps as such) can be represented using SDBM. 714 (56.5\%) of the cases use only SDBM expressions.
In the remaining cases, \meanstddev{90.3\%}{15.9}\footnote{The \meanstddev{\mu}{\sigma} notation indicates the mean and standard deviation.} of expressions and \meanstddev{88.2\%}{17.5} sets can be represented using SDBM.

45 of the test cases perform a total of 7695 emptiness checks.
6262 (81.4\%) of these are performed on HSDBM integer sets, and none on more general SDBMs.
22 (48.9\%) test cases perform emptiness checks only on HSDBM.
In the remaining cases, \meanstddev{73.5\%}{37.7} of the checks are performed on HSDBM sets.

These results suggest that SDBM is sufficient to represent a large fraction of affine constructs appearing in a compiler infrastructure supporting polyhedral compilation~\cite{uday_tensor_cores}, machine learning compilers~\cite{tinyiree}, high-level synthesis~\cite{phism} and other hardware design~\cite{circt}. It is worth noting that the test suite covers rare representational edge cases, so practical applications may have better coverage. For example, many non-SDBM expressions are found in Affine dialect tests, which exercise the full expressive power of (quasi-)affine expressions, including divisions by parameters, huge coefficients, or expressions with hundreds of terms.

Some of the 17 compiled benchmarks consist of multiple translation units processed separately, for a total of 39. Each one was compiled with 7 different configurations, leading to the total of 273 test cases. Out of these, 266 (97.4\%) construct affine expressions and 50 (18.3\%) perform emptiness checks.

\subsubsection{Polygeist.}
96.3\% of the affine expressions and 95.6\% of the integer sets fit the SDBM domain. 185 (69.5\%) cases use only SDBM constructs. The remaining cases have \meanstddev{95\%}{5.1} and \meanstddev{93.8\%}{6.4} SDBM expressions and sets, respectively.

These test cases perform a total of 540 emptiness checks all of which can be expressed using HSDBM. In Polygeist, emptiness checks are performed during dependence analysis. Since the benchmarks are originally written in CUDA, they use only simple single-variable subscript expressions, leading to compatible $i - j$ expressions in dependence relations.

These results indicate that SDBM is suitable for end-to-end compilation, even if a more expressive representation may be occasionally required. Note also the higher ratio compared to the MLIR test suite.

\subsubsection{PPCG.}
While none of the benchmarks can be completely processed using exclusively SDBM, most steps of the compilation process are largely compatible.
Specifically, the initial representation of the program uses only SDBM for 25 (83.3\%) programs, and the result of dependence analysis is representable for 26 (86.7\%) programs.
ILP-based affine scheduling does not match SDBM requirements for any of the programs since it extensively uses multi-variable expressions through its use of the Farkas lemma \cite{10.1145/2429069.2429127}.
On the other hand, the resulting schedule can be expressed as a union of SDBM integer sets for 21 (70\%) programs.
Using the hierarchical form of the schedule~\cite{schedule_trees} instead of a flat union brings this number up to 24 (80\%).
When applying loop tiling on a hierarchical schedule, 23 (76.7\%) programs still use only SDBM with divisibility constraints associated with tile sizes.
Finally, code generation is expressible only for the one program, \texttt{durbin.c},
as it produces linearized expressions of the form $C \cdot i + ii$ to recombine loop indexes after tiling (such linearization was previously avoided in the hierarchical schedule); \texttt{durbin.c} does not contain a tileable loop nest and only accesses single-dimensional arrays with subscripts of the form \texttt{i} and \texttt{i - j - C}, which are SDBM.
We could also confirm our intuition that \emph{all SDBM expressions are also HSDBM}. This is due to congruences being introduced by tiling, which uses the fixed factor of $32$ by default. We verified this by disabling tiling, which brought the number of supported test cases for flat schedule and code generation to 21 (70\%). Tile factors are typically chosen as powers of two or fractions of the problem sizes, so they are likely to remain divisible.

Overall, across all stages and benchmarks, \meanstddev{85.6\%}{21.6} of affine constraints and \meanstddev{78.1\%}{37} sets are SDBM. This number ranges from \meanstddev{41.5\%}{14.3} constraints for the ILP set to \meanstddev{99.8\%}{0.5} for dependency analysis, and from \meanstddev{10.8\%}{24.2} sets for code generation to \meanstddev{99.6\%}{1.1} for dependency analysis.
These results suggest that SDBM combined with structured affine representations such as schedule trees may power a large part of a polyhedral compiler, for all stages except ILP-based affine scheduling.


\subsection{Applications to Translation Validation}
\label{sec:tv}

We additionally used our instrumented version of MLIR\footnote{\texttt{llvmorg-18-init-16246-g4daea501c4fc(Jan 5,2024)}, same for MLIR test suite.} to process three end-to-end machine learning models as described in~\cite{10.1007/978-3-031-13188-2_19}. Specifically, we took the following models (fetched on January 19, 2024).
\begin{itemize}
    \item \texttt{text\_classification\_v2} obtained from \url{https://www.tensorflow.org/lite/examples/text_classification/overview}.
    \item MobileNet v3, variation ``large-075-224-classification'' obtained from \url{https://www.kaggle.com/models/google/mobilenet-v3/frameworks/tfLite}.
    \item SqueezeNet: \url{https://www.kaggle.com/models/tensorflow/squeezenet}.
\end{itemize}

We further converted these models from the original TFLite format into the MLIR TOSA dialect using the following TensorFlow tools:
\texttt{flatbuffer\_translate --tflite-flatbuffer-to-mlir} to yield a TFLite MLIR representation, \texttt{tf-opt --tfl-to-tosa-pipeline} to obtain TOSA.\footnote{Both were compiled from source: \url{https://github.com/tensorflow/tensorflow} version \texttt{ae7eb0931d2973095}, which depends on a different version of MLIR, but the textual representation of TOSA in both is compatible.}
We do not run the models but (partially) compile them along the lines of~\cite{10.1007/978-3-031-13188-2_19}:\footnote{We noticed the existing flag \texttt{tosa-to-linalg-pipeline} does not produce any code, so we reconstructed the MLIR pass pipeline from its source code in \texttt{mlir/lib/Conversion/TosaToLinalg/TosaToLinalgPass.cpp}. Notable differences with the previously reported pipeline include additional TOSA normalization passes and the decomposition of the Standard MLIR dialect into the Arith and Tensor dialects, as well as the recomposition of bufferization passes into a single one.}

\noindent
{\scriptsize
\begin{verbatim}mlir-opt --pass-pipeline='builtin.module(func.func(tosa-optional-decompositions),
  canonicalize, func.func(tosa-infer-shapes, tosa-make-broadcastable, tosa-to-linalg-named),
  canonicalize, func.func(tosa-layerwise-constant-fold, tosa-make-broadcastable),
  tosa-validate, func.func(tosa-to-linalg, tosa-to-arith, tosa-to-tensor),
  linalg-fuse-elementwise-ops, one-shot-bufferize)'
\end{verbatim}}

\noindent
We collected SDBM-related statistics from all three cases in \autoref{tab:model_sdbm}. None of the models required an emptiness check.

\begin{table}[h!tb]
    \vskip-5pt
    \centering
    \begin{tabular}{lrrrrrr}\hline
      \textbf{Model}  & \multicolumn{3}{c}{\textbf{Sets}} & \multicolumn{3}{c}{\textbf{Expressions}} \\
      ~ & \textbf{Total} & \multicolumn{2}{c}{\textbf{SDBM}} & \textbf{Total} & \multicolumn{2}{c}{\textbf{SDBM}}  \\\hline
      Text Classification & 4099 & 4099 & (100\%) & 9148 & 9148 & (100\%) \\
      MobileNet  & 58876 & 52596 & (89.3\%) & 208840 & 202110 & (96.8\%) \\
      SqueezeNet & 28131 & 27806 & (98.8\%) & 96140 & 95490 & (99.3\%) \\\hline
    \end{tabular}
    \smallskip
    \caption{SDBM is sufficient to represent most affine sets and expressions during the partial compilation piepline from TOSA to the bufferized Linalg dialect in MLIR.}
    \label{tab:model_sdbm}
    \vskip-15pt
\end{table}

\section{Related Work}


The relevance of weakly relational domains for loop parallelization and optimization is well established \cite{DBLP:books/mk/AllenK2001}.
More recently, UTVPI approximations enabling complex affine transformations (such as those enabled by PPCG in the empirical evaluation) have also been identified \cite{10.1145/2429069.2429127}.
But these techniques remain unaware of congruence properties, missing optimization opportunities as a result \cite{10.1145/2429069.2429127}.

There is a rich literature on sub-polyhedral domains \cite{10.1145/3457885}.
APRON\footnote{\url{https://antoinemine.github.io/Apron/doc/api/c}} \cite{10.1007/978-3-642-02658-4_52} provides a reference implementation for many of these.
See also
ELINA\footnote{\url{https://elina.ethz.ch}} \cite{10.1145/3009837.3009885} for advanced algorithms and optimizations.
Combinations of abstract domains are popular in static analysis \cite{10.1007/978-3-642-19805-2_31,10.1145/154630.154650}.
These aim at increasing precision by ``cross-fertilization'' of analyses without the need for new abstract domains.
Yet actual intersections of sub-polyhedral domains received much less attention.
Bygde surveys some of these \cite{Bygde948}, the most closely related being the trapezoidal domain \cite{DBLP:conf/ics/Masdupuy92} which combines lattices with intervals, forming a non-relational domain.

Considering SDBM algorithms themselves, our iterative approach to the satisfiability problem is reminiscent of the dynamic all-pairs shortest paths \cite{10.1145/780542.780567} and incremental closure algorithms \cite{HOWE20191}.
Complexity results in this space relate to the cubic upper bound of the Floyd-Warshall algorithm and do not contribute to improving the complexity of the GCD tightening iterations.

The weak NP-completeness of TVPI has been established by Hochbaum and Naor \cite{doi:10.1137/S0097539793251876,hochbaum2004monotonizing,DBLP:journals/networks/Hochbaum21}, together with a (pseudo-polynomial) integer linear programming algorithm that is quadratic in the largest bound of the inequalities.
Our SDBM algorithm has lower complexity and also makes it pseudo-polynomial in the congruence divisors instead.
In compilation problems of interest, congruences correspond to tile and vector sizes dictated by hardware parameters; they are much smaller than bounds of the iteration spaces and arrays.

\section{Conclusion}

We introduced the Strided Difference Bound Matrix (SDBM) abstraction combining two-variable inequalities with congruence constraints.
We demonstrated the prevalence of these across the compiler test suites of MLIR, Polygeist and PPCG.
We showed that the satisfiability of SDBM is NP-hard but also admits an algorithm pseudo-linear in the LCM of the congruence divisors. We identified the Harmonic SDBM (HSDBM) sub-case that commonly arises in compilation problems for deep learning and other areas.
HSDBM satisfiability has a worst-case complexity of $\Ocal(n^4)$, which is practical for uses in compilers and has the potential to accelerate verification tools based on more general Presburger arithmetic. We gave an $\Ocal(mn^4 \log(n\Dlcm))$ algorithm for HSDBM normalization. Finally, given a pair of normalized HSDBM, we showed linear-time algorithms to check for equality and to perform the join operation. The design of a widening operator, also necessary for abstract interpretation, is left for future work.

\begin{credits}

\subsubsection{\discintname}
The authors have no competing interests to declare that are
relevant to the content of this article.
\end{credits}

\newpage

\begin{subappendices}

\renewcommand{\thesection}{\Alph{section}}%

\section{Deferred Proofs}
\label{app:lemmas}

\begin{theorem}
\label{thm:sdbm-sat-np-hard}
The satisfiability problem for simple SDBMs is \NP-hard.
\end{theorem}
\begin{proof}
Let $x$ be a solution to the simple SDBM
\begin{align*}
x_1, \dots x_n \in \Zbb &&
-x_i + x_j \le c_{ij} &&
x_i \equiv 0 \mod d_i.
\end{align*}
Define $y$ by $y_i = x_i/d_i$. Then $x$ is a solution to the above SDBM iff $y$ is a solution to the following Integer Linear Program (ILP):
\begin{align*}
y_1, \dots y_n \in \Zbb &&
d_j x_j - d_i y_i \le c_{ij}.
\end{align*}

We now prove that solving ILPs of the above form is \NP-hard. Lagarias \cite{lagarias1985simultaneous} proved that it is \NP-hard to solve ILPs of the form
\begin{align*}
y, x_1, \dots x_n \in \Zbb
&& -\frac{s}{t} \le \frac{a_i}{b_i}y - x_i \le \frac{s}{t}
\end{align*}
where $s, t, a_i, b_i \in \Zbb$, $s, t, b_i \ge 1$ and $a_i \ne 0$. Note that if some $a_i < 0$, then 
\begin{align*}
    -\frac{s}{t} \le \frac{a_i}{b_i}y - x_i \le \frac{s}{t} \iff
    -\frac{s}{t} \le \frac{|a_i|}{b_i}y - (-x_i) \le \frac{s}{t}.
\end{align*}
Let $z_i = x_i$ if $a_i > 0$ and $z_i = -x_i$ if $a_i < 0$. Then
the system has a solution iff the following system does:
\begin{align*}
y, z_1, \dots z_n \in \Zbb
&& -\frac{s}{t} \le \frac{|a_i|}{b_i}y - z_i \le \frac{s}{t}
\end{align*}
Thus without loss of generality we can assume $a_i \ge 0$ as well. Now we set $A = \lcm(a_1, \dots a_d)$ and multiply each such pair of inequalities by $t b_i A / a_i$. Note that the latter quantity is a positive integer. We then obtain the system
\begin{align*}
y, z_1, \dots z_n \in \Zbb &&
-\frac{s b_i A}{a_i} \le A t y - \frac{A b_i t}{a_i} z_i \le \frac{s b_i A}{a_i}
\end{align*}
which is of the required form, as all variables always get the same coefficient and these coefficients are integers. We reduced the \NP-hard sub-case of ILP to a form equivalent to SDBM satisfiability, showing that the latter is also \NP-hard.
\end{proof}

\lemSdbmSolInnD*
\begin{proof}
Our proof goes along the same lines as the original paper~\cite{cgst1986sensitivity}. Let $y \in S$ and $z \in S \cap L$. Let $A_1$ be the rows of $A$ such that $A_1 z \ge A_1 y$ and let $A_2$ be the remaining rows, so we have $A_2 z < A_2 y$. Let $C = \{x \in \Rbb^n \mid A_1 x \ge 0 \land A_2 x \le 0\}$; clearly $z - y \in C$. Moreover, $C$ is a polyhedral cone. Let $v_1, \dots v_k$ be the generators of $C$; it can be shown that $v_i$ can be chosen such that they have integer elements and each element has absolute value at most $\MASD(A)$~\cite[Lemma 5.4 on page 103]{korte2013combinatorial}. Then $z - y = \sum_i \lambda_i v_i$ for $\lambda \ge 0$.

Let $x_\mu = z - \sum_i \mu_i v_i = y + \sum_i (\lambda_i - \mu_i) v_i$ for $0 \le \mu \le \lambda$; thus there are two equivalent representations for $x_\mu$. $x_\mu$ is a kind of generalized interpolation between $y$ and $z$. We show that $x_\mu \in S$ since $y$ and $z$ are. The first representation shows that $A_1 x_\mu \le A_1 z \le b_1$ where $b_1$ are the constant terms in $b$ for the rows of $A_1$. The second representation shows that $A_2 x_\mu \le A_2 y \le b_2$ where $b_2$ are the corresponding constant terms. Thus, $x_\mu \in S$.

Now set $\nu_i = \lfloor \lambda_i \rfloor_{\Dlcm}$. Then $x_\nu \in S$ as shown above. Moreover, by the first representation, we are just subtracting multiples of $D$ from each element of $z$ so $x_\nu$ satisfies the congruences, and we have $x_\nu \in S \cap L$. Finally, $x_\nu$ is not too far from $y$:
\begin{align*}
\|x_\nu - y\|_\infty &= \Bigg\|\sum_i (\lambda_i - \lfloor
\lambda_i \rfloor_{\Dlcm}) v_i\Bigg\|_\infty \\
&\le \sum_i (\lambda_i - \lfloor \lambda_i \rfloor_{\Dlcm}) \| v_i \|_\infty \\
& \le nD \MASD (A)
\end{align*}
where the last inequality is because we can assume that at most n of the $\lambda_i$ are non-zero by Carathéodory's theorem for conical hulls.
\end{proof}

\lemMASD*
\begin{proof}
Obviously $\MASD(A) \ge 1$ since we can achieve $1$ by just picking any element $\pm 1$ as the submatrix. We now show that $MASD(A) \le 1$. The given matrix can be thought of as an adjacency list of $m$ directed edges in a graph of $n$ vertices. We can assume that the endpoint with coefficient $-1$ is the source and the destination has $+1$. Let $B$ be a square submatrix, which corresponds to a subset of $k$ vertices and $k$ edges.

If some edge $e$ is chosen without any of its endpoints being chosen, the determinant is zero since $e$'s row only has zeros. If exactly one of $e$'s endpoints is chosen, then by Laplace expansion along $e$'s row, we obtain that the absolute determinant of $B$ is equal to that of the matrix obtained by removing $e$ and the one endpoint of it that is present. If by such removals we obtain an empty matrix then $|\det B| = 1$.

Thus we may assume that all edges that are chosen have both their endpoints chosen as well. Consider the undirected version of the graph. Since the number of edges is equal to the number of vertices in the graph, there exists a cycle. Therefore, by adding up the corresponding rows with coefficients $+1$ if the original direction of the edge is used and $-1$ if the reverse direction is used, we can obtain zero. Thus the matrix is not full rank and has determinant zero.
\end{proof}

\begin{lemma}
\label{lem:interval-cap}
Let $I_1, \dots I_n$ be intervals with $I_i = [\ell_i, r_i]$. If $|\cap_i I_i| = k$ for $k \in \Nbb$ then either some interval is of size $k$, or the intersection of some two intervals is of size $k$.
\end{lemma}
\begin{proof}
Note that $\cap_i I_i = [\ell_i, r_j]$ for some $i, j \in [n]$. If $i = j$ then $|I_i| = k$, and otherwise $|I_i \cap I_j| = k$.
\end{proof}

\end{subappendices}

\newpage
\bibliographystyle{splncs04}
\bibliography{references}

\end{document}